\documentclass[twocolumn,10pt]{asme2e}

\confshortname{DSCC2020}
\conffullname{the ASME 2020  \\
              Dynamic Systems and Control Conference}

\confdate{October 4-7}
\confyear{2020}
\confcity{Pittsburgh, PA}
\confcountry{USA}

\papernum{DSCC2020-XXXX}

\title{On the Ergodicity of an Autonomous Robot for Efficient Environment Explorations}

\author{Rabiul Hasan Kabir
    \affiliation{
	Department of Mechanical Engineering\\
	New Mexico Institute of Mining and Technology\\
	Socorro, New Mexico, 87801\\
    Email: rabiul.kabir@student.nmt.edu
    }	
}

\author{Kooktae Lee\thanks{Address all correspondence to this author.}
    \affiliation{
    Department of Mechanical Engineering\\
	New Mexico Institute of Mining and Technology\\
	Socorro, New Mexico, 87801\\
    Email: kooktae.lee@nmt.edu
    }	
}

\usepackage{amsmath,graphicx,amsfonts,amssymb,epsfig,subfig,mathrsfs,mathtools}
\usepackage{amsmath,lipsum}
\usepackage{algorithm}
\usepackage{algpseudocode}
\usepackage{graphicx}
\usepackage{xcolor}

\newtheorem{assumption}{Assumption}
\newtheorem{theorem}{Theorem}

\newtheorem{proposition}{Proposition}

\allowdisplaybreaks

\begin{document}

\maketitle    

\begin{abstract}
{\it This paper addresses the autonomous robot ergodicity problem for efficient environment exploration. The spatial distribution as a reference distribution is given by a mixture of Gaussian and the mass generation of the robot is assumed to be skinny Gaussian. The main problem to solve is then to find out proper timing for the robot to visit as well as leave each component-wise Gaussian for the purpose of achieving the ergodicity. 
The novelty of the proposed method is that no approximation is required for the developed method.
Given the definition of the ergodic function, a convergence condition is derived based on the timing analysis. Also, a formal algorithm to achieve the ergodicity is provided.
To support the validity of the proposed algorithm, simulation results are provided.}
\end{abstract}

\begin{nomenclature}
\entry{$\mathbb{R}$}{A set of real numbers}
\entry{$\mathbb{N}$}{A set of natural numbers}
\entry{$\mathbb{N}_0$}{A set of non-negative integers}
\entry{$k$}{A discrete-time index such that $k\in\mathbb{N}_0$}
\entry{$^{T}$}{A transpose operator}
\entry{$X$}{A domain such that $X\in\mathbb{R}^2$}
\entry{$X(\cdot)$}{A subset of the domain $X$ for a given function}
\entry{$\bigcup$}{A union operator for given sets}
\entry{$\mathcal{N}(\mu, \Sigma)$}{A Gaussian distribution with a mean $\mu$ and a covariance $\Sigma$}
\end{nomenclature}
\section*{INTRODUCTION}
Recently, the ergodic environment exploration scheme for autonomous robots has attracted many attentions because of the exploration efficiency as well as its wide applicability. The efficiency in this context implies that the robot explores an environment such that the distribution from the time-averaged robot trajectories is the same as the given reference distribution. As such, the robot can efficiently cover an environment with some priority associated with the given reference distribution.
This ergodic exploration scheme can be employed for various missions including search and rescue, surveillance and reconnaissance, site inspection, wildlife monitoring, space exploration, etc. 

The first attempt to employ the concept for the ergodicity in autonomous agents is introduced in \cite{mathew2011metrics}. In this study, a new method is provided to measure the ergodicity of agents compared to a given probability measure. This metric, based on Fourier Basis Function, is developed for centralized feedback control laws applicable in multi-agent systems. 
An algorithm for determining the optimal trajectories for autonomous robots is designed in \cite{silverman2013optimal} for the purpose of data acquisition. The target of this study is to design an automated trajectory using optimal transport so that the robot spends more time on the regions where there is a higher probability of getting informative data and less time where the probability of getting information is lower. Fourier basis function based ergodic metric has been used to calculate ergodicity of the robot trajectories. 
A similar concept is investigated by \cite{miller2013trajectory} with an algorithm that generates trajectories with a goal of exploring a region efficiently while considering a probabilistic information density representation of that region. The problem has been defined as a continuous time trajectory optimization problem and the objective function requires the correlation between the spatial probability distribution and time-averaged trajectory. General nonlinear robot dynamics has been considered in this study.
An extension of ergodic area coverage algorithm is presented in \cite{ayvali2017ergodic} for multiple robots working in constrained environment, where there are presence of obstacles and restricted areas. The coordination of multiple agents with various sensing abilities were used for demonstrating ergodic coverage of a domain. 
The study \cite{o2015optimal} proposes an algorithm termed Ergodic Environmental Exploration (E3), a finite receding horizon optimal control algorithm, for the purpose of exploration of an unknown environment that includes regions with varying degrees of importance. This algorithm helps to ensure minimum control effort and minimum difference between time average behavior of system’s trajectory and distribution of the information gain. Experiments have been conducted on robots using this algorithm and results have been presented in this study. 
Also, an iterative optimal control algorithm for general nonlinear dynamics is proposed in \cite{prabhakar2015symplectic}. The metric for information gain is the difference between the spatial distribution and the statistical time-averaged trajectory. Two discrete-time iterative optimization approaches have been demonstrated in this study -- first order discretization and symplectic integration. The authors presented that discretization choice for a system has significant effect on control and state trajectories. 
In \cite{mavrommati2017real}, the authors have developed a receding horizon ergodic control approach and their nonlinear model predictive control algorithm improves the ergodicity between an information density distribution of the sensor domain and real time motion of agents. This approach allows the agents to perform independently and to share information regarding their coverage across a communication network. 
In \cite{de2016ergodic}, a trajectory optimization approach is developed for robotic ergodic exploration where stochastic nonlinear sensor dynamics has been considered. A new approach is introduced in this study and the provided results show that the developed algorithm can generate trajectories that can ensure greater and more predictable ergodicity. A decentralized ergodic control strategy is proposed in \cite{abraham2018decentralized} for multi-agent systems with nonlinear dynamics. The agents only need to share a coefficient related to the action of each agent with each other to make decentralized decisions. 

However, all of the aforementioned researches heavily rely on the ergodic metric defined in \cite{mathew2011metrics}, which employed the Fourier basis function to obtain the distribution for the time-averaged robot trajectories. This Fourier basis function intrinsically entails an approximation during implementation as it has an infinite summation term.
Although a similar idea related to the ergodicity is proposed in \cite{milutinovi2006modeling,hamann2008framework,qi2014multi,ivic2016ergodicity,eren2017velocity} based on the global behaviors of multiple agents using the macrostate of the partial differential equation, the desired behavior is only achieved when the number of agents are extremely large.
Also, a new approach for the ergodic exploration plan is proposed in \cite{kabir2020ACC} based on the optimal transport theory \cite{lee2014probabilistic,lee2014optimal,lee2015performance,lee2018optimal}, however, this method includes an approximation caused by sampling representation of the given spatial distribution.

In this paper, we propose a new approach to realize robot ergodic explorations based on the timing analysis. The spatial distribution as a reference is given as a mixture of Gaussian. Then, the generation of a mass by the robot is assumed to be skinny Gaussian distribution. The problem addressed here is to find out the proper timing for the robot to visit as well as leave each component-wise Gaussian for the purpose of achieving the ergodicity. The major contribution of this study is that unlike other researches that employed Fourier basis function, which necessarily entails an approximation error, the proposed method does not include any approximations. The convergence condition is derived based on the defined ergodic function, to achieve the robot ergodicity.
To support the technical soundness of the proposed method, simulation results are provided.

\section*{PROBLEM DESCRIPTION}
This section addresses the problem for realizing the ergoricity, followed by the formulation of the problem. Throughout the paper, the spatial distribution $\rho^*$ is given as a reference distribution and is assumed to have the following property.
\begin{assumption}
The given spatial distribution $\rho^*$ is expressed as a Mixture of Gaussian (MoG) in the following form:
\begin{align}
    \rho^* = \sum_{i=1}^{m}\alpha_i\mathcal{N}(\mu_i, \Sigma_i),
\end{align}
where $\alpha_i$ is a weight such that $0< \alpha_i < 1$, $\forall i$ with $\sum_{i=1}^{m}\alpha_i=1$ and $\mathcal{N}(\mu_i, \Sigma_i)$ is a Gaussian distribution with a mean $\mu_i$ and a covariance $\Sigma_i$.
\end{assumption}
Although it is not necessary, $\rho^*$ is assumed to be stationary for simplicity. 

The robot generates a unit mass concentrated on the current robot position $\mu_k^R$ with a form of a skinny Gaussian distribution as in the following assumption.
\begin{assumption}
Suppose that the robot position at any discrete time $k$ is given as $\mu_k^R$. At each time step, the mass generated by the robot is represented by a skinny Gaussian $f_k:=\mathcal{N}(\mu_k^R,\Sigma^R)$, where $\Sigma^R$ is stationary and is given such that its distribution is narrow.
\end{assumption}
Mathematically, $f_k$ for the two dimensional case has the following structure:
\begin{equation}\label{eqn: f_{k}}
    \begin{aligned}
        f_{k} &= \frac{1}{\sqrt{(2\pi)^2\vert\Sigma^R\vert}}\exp\left(-\frac{1}{2}(x-\mu_{k}^R)^T(\Sigma^R)^{-1}(x-\mu_{k}^R)\right)\\ 
\end{aligned}
\end{equation}
where $\vert\cdot\vert$ is the determinant and $\mu_k^R$ represents the current robot position.

The proposed method to realize the ergodicity is based on the discrete-time dynamics, however, the time-averaged distribution for the continuous time case is given below to provide better description.
\begin{equation}\label{eqn: rho_cont}
\rho(x,t) = \frac{1}{t}\int f(x,t)dt,
\end{equation}
where $f(x,t)$ denotes a skinny Gaussian in the continuous time case.

Notice that in the above equation, the integral is taken with respect to time and then, it is divided by the total elapsed time, which is to represent the time-averaged behavior of the robot.
The counterpart corresponding to the discrete-time case is then written by
\begin{equation}\label{eqn: rho_dis}
    \rho_{k} := \rho(x,k) =  \frac{1}{k+1}(\sum_{i=0}^{k} f_i),
\end{equation}
where $f_i$ stands for a skinny Gaussian in the discrete-time step. 

The difference between the time-averaged distribution $\rho_k$ and the given spatial distribution $\rho^*$ at time $k$ is written as
\begin{align}\label{eqn: phi_{k}}
    \phi_k=\rho_k-\rho^*
\end{align}
Further, the ergodic function $V_k$ is defined as the integral of the absolute value of $\phi_k$ over the given domain by
\begin{align}
    V_k = \int_{\Omega}|\phi_k|dx\label{eqn: V_k}
\end{align}
Notice that $V_k$ always ranges between $0$ and $2$, regardless of $\rho_k$ and $\rho^*$ due to its mathematical definition. For instance, $V_k=2$ if $\rho_k$ is accumulated completely outside the domain of $\rho^*$.

\begin{figure}[!tbph]
\begin{center}
\includegraphics[scale=0.32]{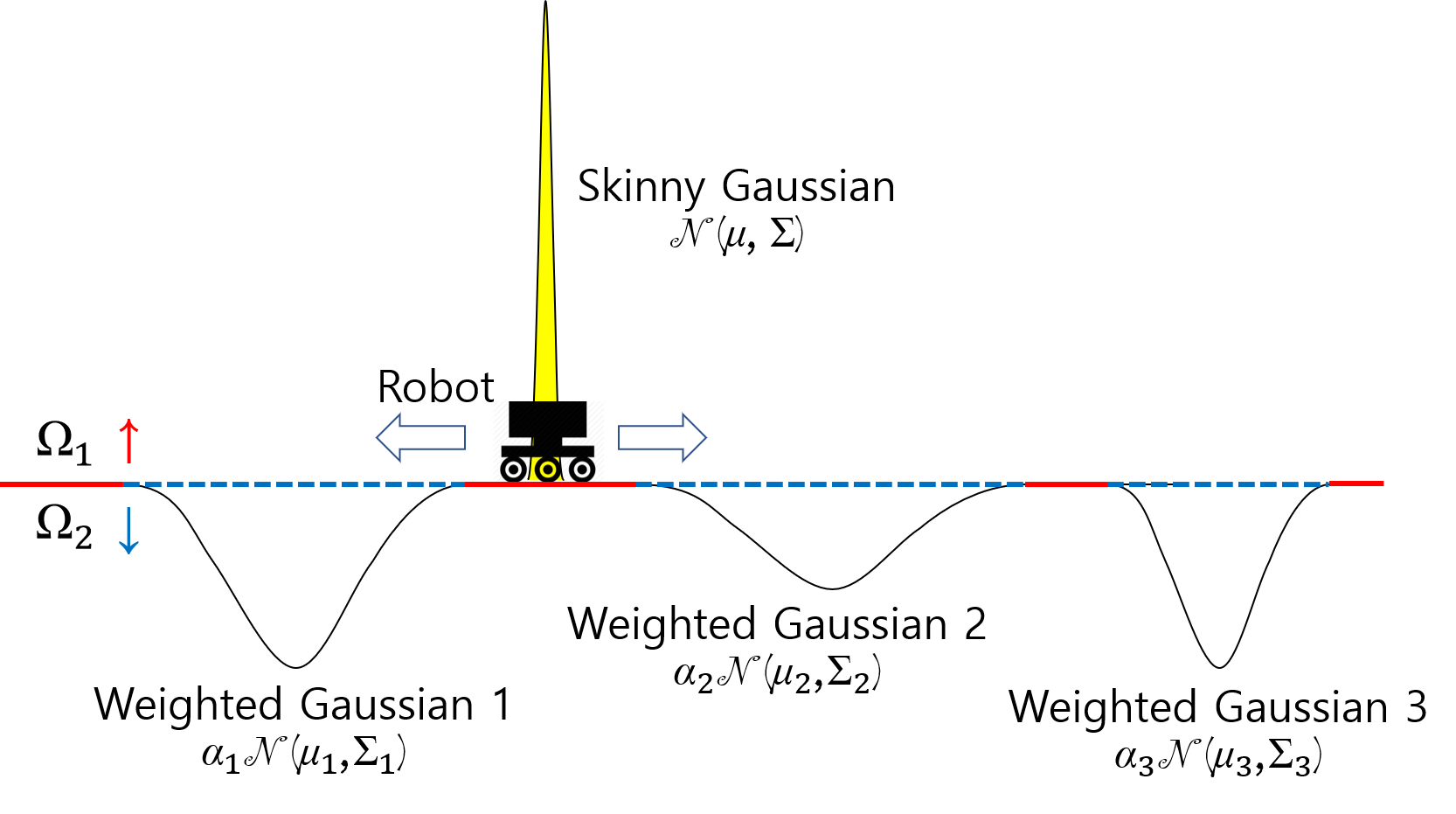}
\caption{\uppercase{Schematic of the robot ergodic trajectory generation problem }}\label{fig: problem_schematic}
\end{center}
\end{figure}

In Fig. \ref{fig: problem_schematic}, we illustrate the robot with the skinny Gaussian mass generation and the given spatial distribution being as an MoG with a negative sign, and hence below the zero base line.
We define this portion (below the zero line) as a hole of which domain is denoted by $\Omega_2$\ (blue dashed lines in Fig. \ref{fig: problem_schematic}). The remaining region outside $\Omega_2$ is represented by $\Omega_1$ (red solid lines in Fig. \ref{fig: problem_schematic}). Alternatively, $\Omega_1$ and $\Omega_2$ are defined as
\begin{align*}
    \Omega_1 &:= X - \Omega_2, \qquad
    \Omega_2 := \{x\vert x\in \bigcup_{i=1}^{m} X(\mathcal{N}\left(\mu_i, \Sigma_i)\right) \},
\end{align*}
where $X(\mathcal{N}\left(\mu_i, \Sigma_i)\right)$ denotes the domain belongs to the Gaussian $\mathcal{N}\left(\mu_i, \Sigma_i\right)$.

The main goal of this paper is to achieve the ergodicity such that $V_k\rightarrow 0$ as $k\rightarrow \infty$, meaning the time-averaged distribution from the robot trajectories converges to the given spatial distribution. One may infer that this goal is achieved by making the robot to stay at each hole (or component-wise weighted Gaussian in a given MoG as shown in Fig. \ref{fig: problem_schematic}) with a given portion $\alpha_i$. However, it is not as simple as it would be because of the following reasons. Recalling the time-average dynamics in \eqref{eqn: rho_dis}, it can be rewritten recursively by
\begin{align}
    \rho_{k} = \dfrac{1}{k+1}\left(k\cdot\rho_{k-1} + f_{k}\right)\label{eqn: f_k recursive}
\end{align}
According to \eqref{eqn: f_k recursive}, the contribution of the current mass generation $f_k$ to the time-averaged distribution $\rho_k$ reduces nonlinearly by $\dfrac{1}{k+1}$, which induces the difficulty in attaining the ergodicity. 
Secondly, even if one hole is completed filled with a mass generated by the robot, the mass vanishes gradually as soon as the robot leaves that hole. Finally, the robot cannot jump from one hole to another and hence, it spills unnecessary masses while traversing the $\Omega_1$ region.
As such, it is not clear what is the timing for the robot to visit each hole and how long it should stay there.
In what follows, we thus provide the analysis to guarantee that $V_k$ is decreasing under a certain condition.

\section*{ERROR ANALYSIS OF ERGODIC OPERATION}
Throughout the paper, the variable $h$ is given to denote the time steps for the robot being inside $\Omega_1$. Similarly, $h'$ indicates the time steps in a hole to explore that hole.
In the multiple holes case, a subscript will be used to stands for a specific hole number.
Before proceeding to the error analysis, the following proposition sheds light on how the time-averaged distribution changes as the robot moves in the domain.

\begin{proposition}\label{proposition: delta_rho_k}
Given the time-averaged distribution $\rho_k$ at any time $k$, the variation in the time-averaged distribution after $h$ time steps, $\Delta \rho_k^h$, can be calculated by the following equation:
\begin{align*}
       \Delta \rho_k^h&= \rho_{k+h}-\rho_k=\frac{1}{k+h+1}\left(\sum_{i=k+1}^{k+h}f_i-h\rho_k \right )
\end{align*}
\end{proposition}
\begin{proof}

From \eqref{eqn: rho_dis}, the time-averaged distribution at $k+h$ can be written as
\begin{align}\label{eqn: rho_k+h}
    \nonumber
    \rho_{k+h} &= \frac{1}{k+h+1}(\sum_{i=0}^{k+h} f_i)\\
    &= \frac{1}{k+h+1}\left( (k+1)\rho_k +\sum_{i=k+1}^{k+h} f_i  \right)
\end{align}

where,
\begin{align}
    \qquad\qquad\qquad\qquad(k+1)\rho_k=\sum_{i=0}^{k}f_i\nonumber \qquad\qquad\qquad\text{From \eqref{eqn: rho_dis}}
\end{align}

We can rewrite \eqref{eqn: rho_k+h} as
\begin{align}
     (k+h+1)\rho_{k+h}-(k+1)\rho_k = \sum_{i=k+1}^{k+h} f_i
\end{align}

Finally, the following expression can be obtained for $\Delta \rho_k^h$ from the previous equation.
\begin{align}
    \nonumber
    \Delta\rho_k^h  = \rho_{k+h}-\rho_k= \frac{1}{k+h+1}\left(\sum_{i=k+1}^{k+h} f_i - h\rho_k \right)
\end{align}


\end{proof}

\begin{figure}
\begin{center}
\includegraphics[scale=0.45]{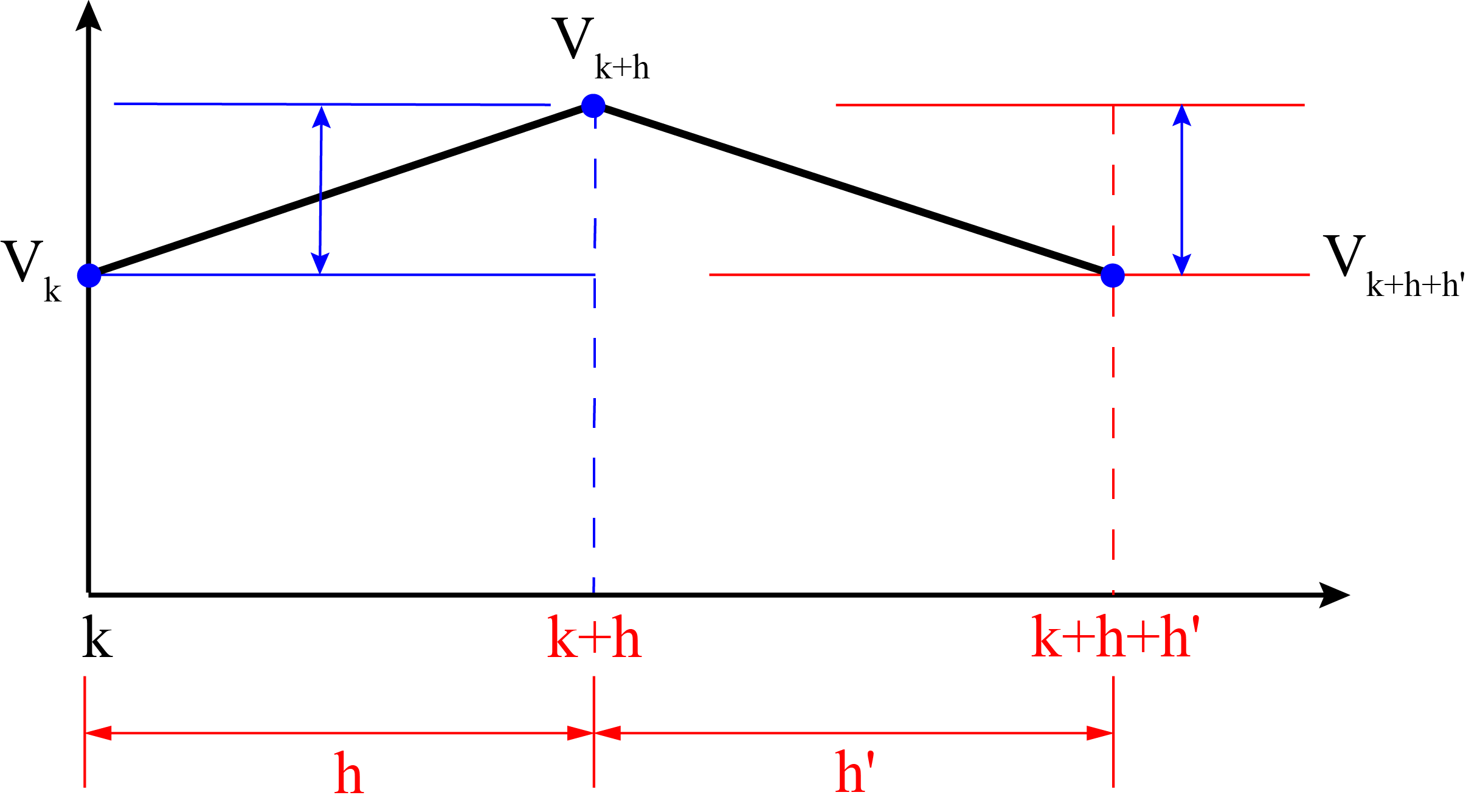}
\caption{\uppercase{Piece-wise variation of ergodic function with discrete time }}\label{fig: V_plot}
\end{center}
\end{figure}

 Travelling the $\Omega_1$ region always increase $V_k$ as the robot is spending time in the area where it should not be. On the other hand, as the robot explores a hole to match with the given spatial distribution for that hole, $V_k$ goes down. 
 The proposed strategy to achieve ergodicity is explained in the following way.
 If the robot spends $h$ time steps in $\Omega_1$ region, then $V_k$ increases by a certain amount, which is illustrated in Fig. \ref{fig: V_plot}. To guarantee the piece-wise decreasing property for $V_k$ at time $k+h+h'$, the decrement from $V_{k+h}$ to $V_{k+h+h'}$ should be greater than the increment from $V_k$ to $V_{k+h}$ as shown in Fig. \ref{fig: V_plot}.
If this conditions is satisfied throughout the robot explorations, then the ergodic function $V_k$ will converge to zero, which is defined as a piece-wise convergence. Therefore, this condition is provided in the following theorem, developed for the piece-wise convergence of the ergodic function. 
 
\begin{theorem}\label{theorem: 1}
Consider the addressed robot ergodicity problem to realize $V_k\rightarrow 0$ as $k\rightarrow \infty$. 
Given $h$ time steps for the robot in $\Omega_1$ while reaching a certain hole, the ergodic function $V_k$ is a piece-wise contraction mapping, if the robot stays at the hole for $h'$ time steps, expressed in the following form:
\begin{align}
     h' > \left( \frac{\int_{\Omega_2}\rho_k dx}{\int_{\Omega_1}\rho_k dx} \right )h\label{eqn: theorem_1}
\end{align}
\end{theorem}
In this case, the following property
\begin{align}
\nonumber
    |V_{k+h+h'}-V_{k+h}|>|V_{k+h}-V_k|
\end{align}
is satisfied.
\begin{proof}
From Fig. \ref{fig: V_plot}, it is shown that for $h$ time steps, the ergodic function $V_k$ goes up such that $V_{k+h}>V_k$ and hence, $|V_{k+h}-V_k|=V_{k+h}-V_k$. The expression for the change of $V$ for $h$ time steps, $|V_{k+h}-V_k|$ can be derived from the following calculation: 
\begin{align}\label{eqn: V for h (1)}
|V_{k+h}-V_k|&=\int_{\Omega_1}(|\phi_{k+h}|-|\phi_k|)dx+\int_{\Omega_2}(|\phi_{k+h}|-|\phi_k|)dx 
\end{align}

For an ideal case, $\phi_k$ is always negative in $\Omega_2$ and positive in $\Omega_1$. Based on this observation, 
\eqref{eqn: V for h (1)} can be rewritten by replacing $\phi_k$ in \eqref{eqn: V for h (1)} with \eqref{eqn: phi_{k}} as
\begin{align}
\nonumber
   |V_{k+h}-V_k|&=\int_{\Omega_1}(\rho_{k+h}-\rho_k)dx-\int_{\Omega_2}(\rho_{k+h}-\rho_k)dx\\\nonumber
   & = \int_{\Omega_1}\Delta\rho_{k}^hdx-\int_{\Omega_2}\Delta\rho_{k}^hdx
\end{align}

Utilizing the result in Proposition  \ref{proposition: delta_rho_k}, the above equation can be further expressed by
\begin{align}\label{eqn: V for h (2)}
   |V_{k+h}-V_k|&=\frac{1}{k+h+1}\left[ \int_{\Omega_1} \left (\sum_{i=k+1}^{k+h}f_i-h \rho_k \right )dx  + \int_{\Omega_2} h\rho_k dx  \right] 
\end{align}

Now, we have
\begin{align}\label{eqn: h_omega1}
\nonumber
    -\int_{\Omega_1}h\rho_k dx + \int_{\Omega_2}h\rho_k dx &= -2h\int_{\Omega_1}\rho_k dx + h\int_{\Omega_1+\Omega_2}\rho_k dx\\ &= -2h\int_{\Omega_1}\rho_k dx + h
\end{align}

and
\begin{align}\label{eqn: f_i omega1}
    \int_{\Omega_1} \sum_{i=k+1}^{k+h}f_i dx = \sum_{i=k+1}^{k+h}\int_{\Omega_1}f_i dx = \sum_{i=k+1}^{k+h}1 = h
\end{align}

Using \eqref{eqn: h_omega1} and \eqref{eqn: f_i omega1}, \eqref{eqn: V for h (2)} can be written as:
\begin{align}\label{eqn: V for h (3)}
   |V_{k+h}-V_k|=\frac{2h\int_{\Omega_2}\rho_k dx}{k+h+1}
\end{align}
where, 
\begin{align}
\nonumber
   \int_{\Omega_2}\rho_k dx = 1-\int_{\Omega_1}\rho_k dx 
\end{align}


The next step is to derive $h'$ such that $V_{k+h+h'}<V_{k+h}$. From Fig. \ref{fig: V_plot}, it can be observed that for $h'$ time steps, the ergodic function $V_k$ decreases. 
In this case, it satisfies $|V_{k+h+h'}-V_{k+h}|= -(V_{k+h+h'}-V_{k+h})$. Then, the expression for $|V_{k+h+h'}-V_{k+h}|$ can be obtained from the following calculation:
\begin{align}\label{eqn: V for h' (1)}
    & \phi \downarrow: \quad |V_{k+h+h'}-V_{k+h}|\\\nonumber
    &=-\left(\int_{\Omega_1}(|\phi_{k+h+h'}|-|\phi_{k+h}|)dx+\int_{\Omega_2}(|\phi_{k+h+h'}|-|\phi_{k+h}|)dx \right )
\end{align}

It has been mentioned before that $\phi$ is always negative in $\Omega_2$ and positive in $\Omega_1$. Using this observation and replacing $\phi$ by its expression from  \eqref{eqn: phi_{k}}, we can rewrite \eqref{eqn: V for h' (1)} as,
\begin{align}\label{eqn: V for h' (2)}
\nonumber
    |V_{k+h+h'}-V_{k+h}|&=-\left(\int_{\Omega_1}(\rho_{k+h+h'}-\rho_{k+h})dx\right.\\\nonumber
    & \qquad\left.-\int_{\Omega_2}(\rho_{k+h+h'}-\rho_{k+h})dx \right )\\
    & = -\left( \int_{\Omega_1}\Delta\rho_{k+h}^{h'}dx - \int_{\Omega_2}\Delta\rho_{k+h}^{h'}dx \right ) 
\end{align}

Again, Proposition 1 for $k+h+1$ and $k+h+h'+1$ can be expressed as :
\begin{align}\label{eqn: delta rho h'}
    \nonumber
       \Delta \rho_{k+h}^{h'}&= \rho_{k+h+h'}-\rho_{k+h}\\
       &=\frac{1}{k+h+h'+1}\left(\sum_{i=k+h+1}^{k+h+h'}f_i-(h')\rho_{k+h} \right )
\end{align}

By substituting terms in \eqref{eqn: V for h' (2)} by \eqref{eqn: delta rho h'}, it further leads to
\begin{align}\label{eqn: V for h' (3)}
\nonumber
    |V_{k+h+h'}-V_{k+h}|&=-\frac{1}{k+h+h'+1}\left[\int_{\Omega_1}-h'\rho_{k+h}dx\right.\\
    & \qquad\left.-\left(\int_{\Omega_2}\sum_{i=k+h+1}^{k+h+h'}f_i -h'\rho_{k+h}dx \right ) \right ]
\end{align}

Similar to \eqref{eqn: h_omega1} and \eqref{eqn: f_i omega1}, the following equations are obtained:
\begin{align}\label{eqn: h_omega2}
    -\int_{\Omega_1}h'\rho_{k+h} dx + \int_{\Omega_2}h'\rho_{k+h} dx= -2h'\int_{\Omega_1}\rho_{k+h} dx + h'
\end{align}
and
\begin{align}\label{eqn: f_i omega2}
    \int_{\Omega_2} \sum_{i=k+h+1}^{k+h+h'}f_i dx = h'
\end{align}

Applying \eqref{eqn: h_omega2} and \eqref{eqn: f_i omega2} to \eqref{eqn: V for h' (3)} results in
\begin{align}\label{eqn: V for h' (4)}
    |V_{k+h+h'}-V_{k+h}|=\frac{2h'\left( 1-\int_{\Omega_2}\rho_{k+h}\right )dx}{k+h+h'+1}
\end{align}
where,
\begin{align}\label{eqn: 1-int rho_k_h}
    1-\int_{\Omega_2}\rho_{k+h}dx = \int_{\Omega_1}\rho_{k+h}dx
\end{align}


To ensure the piece-wise convergence, the following condition must be satisfied:
\begin{equation}\label{eqn: V ineq}
    |V_{k+h+h'}-V_{k+h}|>|V_{k+h}-V_k|
\end{equation}

We can write \eqref{eqn: 1-int rho_k_h} in terms of $k$, $h$ and $\rho_k$ from 
\begin{align}
    \nonumber
     \int_{\Omega_2}\Delta\rho_k^hdx &= \int_{\Omega_2}(\rho_{k+h}-\rho_k)dx=-\frac{h}{k+h+1}\int_{\Omega_2}\rho_k dx\\\nonumber
\int_{\Omega_2}\rho_{k+h}dx &= \frac{k+1}{k+h+1}\int_{\Omega_2}\rho_k dx\\\nonumber
 1-\int_{\Omega_2}\rho_{k+h}dx &=1-\frac{k+1}{k+h+1}\int_{\Omega_2}\rho_k dx = 1- \left(\frac{k+1}{k+h+1} \right )a \nonumber
\end{align}
where,
\begin{align*}
    a = \int_{\Omega_2}\rho_k dx
\end{align*}

Finally, the condition ensuring the validity of \eqref{eqn: V ineq} is then calculated by
\begin{align}\label{eqn: V_h'>V_h}
    \nonumber
    &|V_{k+h+h'}-V_{k+h}|>|V_{k+h}-V_k|\\\nonumber
    \Rightarrow \quad &\frac{2h'\left(1-\left(\frac{k+1}{k+h+1}\right )a \right )}{k+h+h'+1}>\frac{2ha}{k+h+1}\\
    \Rightarrow \quad &h'>\frac{h(k+h+1)a}{(k+h+1)(1-a)}=\frac{ha}{(1-a)}
\end{align}

or alternatively,
\begin{equation*}
     h' > \left( \frac{\int_{\Omega_2}\rho_k dx}{\int_{\Omega_1}\rho_k dx} \right )h
\end{equation*}
\end{proof}

Theorem \ref{theorem: 1} indicates how much time, $h'$, the robot should stay at a certain hole when the robot travels in $\Omega_1$ with $h$ amounts of time steps. Once satisfied, this condition guarantees that the ergodic function $V_k$ will be piece-wise decreasing. 
In the sequel, a formal algorithm is presented to provide the rule for hole departure timing  as well as the robot position update law.
\section*{ALGORITHM}

\begin{figure}
\begin{center}
\includegraphics[scale=0.45]{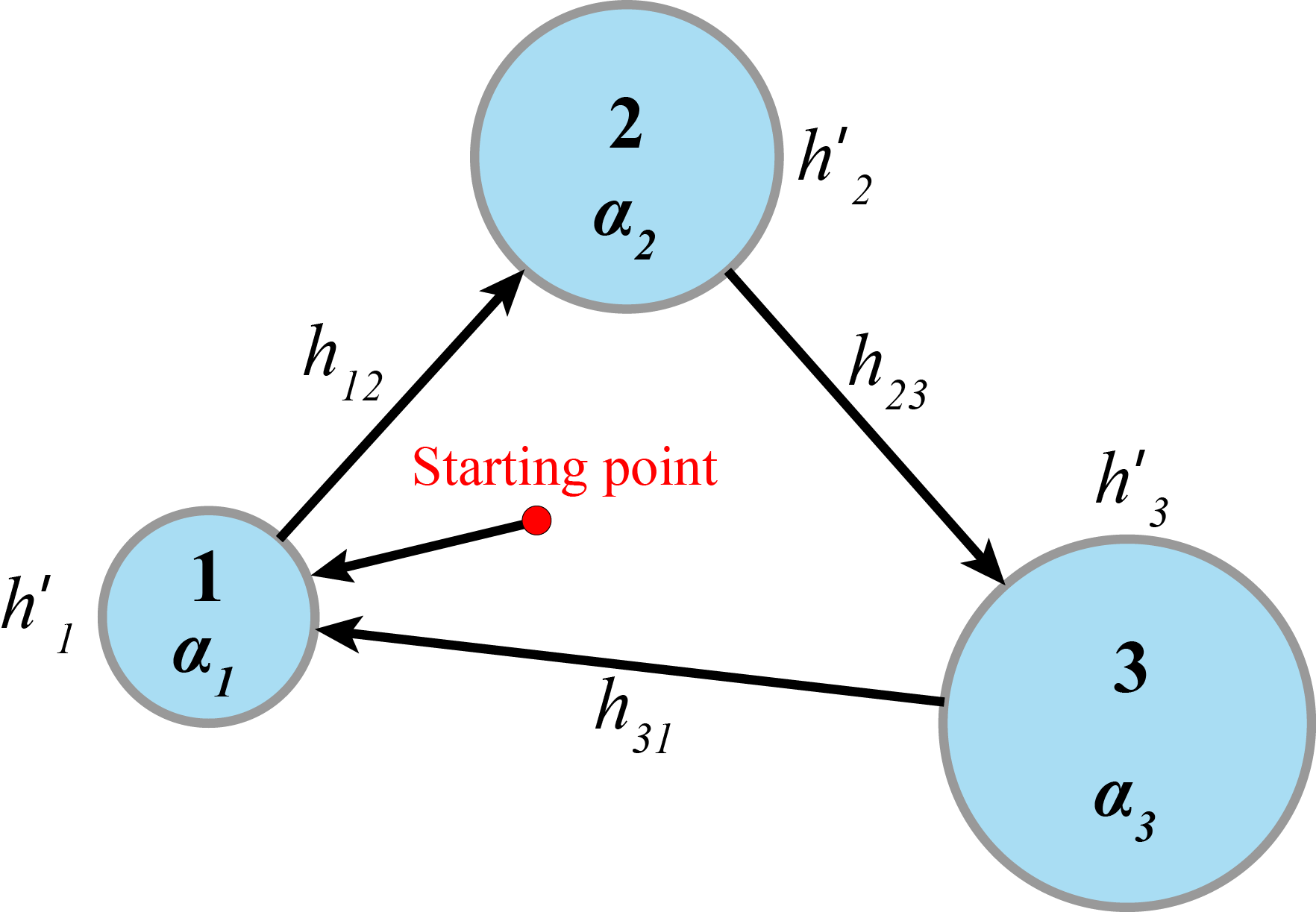}
\caption{\uppercase{Ergodic exploration trajectory of a robot for 3-hole spatial distribution}}\label{fig: Algorithm}
\end{center}
\end{figure}
The formal algorithm to achieve the ergodicity is provided in this section.
Fig. \ref{fig: Algorithm} illustrates how the robot traverses in the domain $\Omega$. The red point in the figure is given as a starting point for the robot. At this moment, the robot searches for the nearest hole as a target hole (e.g., the hole of which 3-sigma boundary is the closest to the robot position). Once selected, the robot moves toward the minimum point in that hole.
The next robot position $\mu_{k+1}^R$ is updated using the following equation:
\begin{align}\label{eqn: robot dynamics}
    \mu_{k+1}^{R} = \mu_{k}^R + v_{\max}\cdot\dfrac{g^k-\mu_{k}^R}{\lVert g^k-\mu_{k}^R \rVert}
\end{align}
where $v_{max}$ is the maximum velocity attainable by the robot and $g^k$ denotes the current minimum value point in the hole as a goal position.\par
Notice that the robot dynamics is not considered here since it is out of scope of this paper. Rather, this study provides the timing for a robot to stay each hole for the realization of the ergodicity. However, one may consider different robot dynamics to update the next robot position $\mu_{k+1}^R$.\par

For the given example in Fig. \ref{fig: Algorithm}, it is clear that hole $1$ is the closest to the robot initial position. Therefore, the robot needs to explore hole $1$ first. 
The robot creates a mass with skinny Gaussian distribution in every time steps, as described in Assumption 2. The robot always determine the location where the minima of $\phi_k$ exists, sets it as the current goal point and updates the goal in every time step while travelling in hole $1$ and moves toward the current updated goal point. 

Although the required staying time in the current hole is proposed in \eqref{eqn: theorem_1} for the convergence of $V_k$, it only provides the lower bound. This implies that the convergence speed for $V_k$ may be too slow if the robot leaves a hole as soon as \eqref{eqn: theorem_1} is satisfied.
The following condition is thus given to provide a proper departure time from the current hole:
\begin{align}\label{eqn: beta}
    h'':=\int_{\Omega_2\cap X(\text{target hole})} \vert \phi_k\vert dx > c_N
\end{align}
where $c_N = \beta \cdot e^{-\gamma\cdot N}$ with 
$\beta$ and $\gamma$ being some positive coefficients and $N$ as a cycle number. This cycle number $N$ increases when the robot visited all holes and arrives at the first visited hole again.

This condition ensures that the robot should stay until the accumulated error $\int_{\Omega_2\cap X(\text{hole 1})} \vert \phi_k\vert dx$ in the current hole becomes greater than $c_N$. In other words, the hole need to be filled by the certain amount defined by \eqref{eqn: beta}.  The reason behind $c_N$ given in the above form instead of zero is as follows.
Firstly, the $\frac{1}{k+1}$ term in \eqref{eqn: rho_dis} indicates that at the beginning of the exploration, the generated mass $f_k$ has greater impact on $\phi_k$ as $\frac{1}{k+1}$ is relatively high. 
Thus, $V_k$ may increase even though the robot is actually filling the hole. Secondly, the robot cannot leave the hole instantaneously when it decides to do so, resulting in some extra mass added in the hole. As a result, the error $\phi_k$ may have a positive value in the hole if $c_N$ is zero, which is not desirable. Therefore, $c_N$ is given in \eqref{eqn: beta} such that at the initial stage, the robot decides to exit the hole when there still exists some error in the hole, and as time increases the robot decides to leave the hole with less and less error in it during later explorations. \par

If the robot staying time in the current hole is greater than $\max(\Bar{h}', \Bar{h}'')$, where  $\Bar{h}'$ and $\Bar{h}''$ are defined by the time when it first satisfies the condition \eqref{eqn: theorem_1} and \eqref{eqn: beta}, respectively, then the robot moves toward the next hole. This next hole is predetermined by the given configuration of an MoG, to connect each hole with the shortest path as shown in Fig. \ref{fig: Algorithm}.  In this way, it is guaranteed that the robot fills the hole such that $V_k$ is contracting with relatively fast convergence speed.
Again in Fig. \ref{fig: Algorithm}, hole $2$ should be the next hole instead of $3$ as it is closer to hole $1$. 
The robot travels to the hole $2$ with a goal position $g^k$ set as the minima in hole $2$. If the robot fills the hole and the given staying condition is greater than $\min(\Bar{h}',\Bar{h}'')$, then the robot moves toward hole $3$.
In traversing $\Omega_1$, the robot may not take the same trajectory from one hole to another since $h_{12}$, $h_{23}$ and $h_{31}$ vary from different explorations, where $h_{ij}$ is the time spent outside the holes for reaching hole $j$ from $i$. \par

A pseudo code is provided below to illustrate the formal procedure of the proposed ergodic algorithm.

\begin{algorithm}[!h]
\caption{Ergodic Exploration Algorithm}\label{algorithm:1}
\begin{algorithmic}[1]
\State initialize $\rho^*$, $v_{max}$, $f_0$, $k\gets 0$
\State Find the target hole:
\If{$k=0$} it is given as the closest hole
\Else{ it is updated by the given configuration of an MoG}
\EndIf
\While{the robot staying time in the current hole $<\max(\Bar{h}',\Bar{h}'')$}
\State Find the minimum location $\min(\phi_k)$ in the target hole and set it as the current goal point $g^k$
\State Update the next robot position $\mu_{k+1}$ by \eqref{eqn: robot dynamics}
\State Fill up the target hole by generating a mass $f_k$
\State Update $\rho_k$ from \eqref{eqn: rho_dis}
\State Calculate $\phi_k$ and $V_k$ from  \eqref{eqn: phi_{k}} and \eqref{eqn: V_k}, respectively
\State $k\gets k+1$
\EndWhile
\State Repeat from step 2 for the next hole
\end{algorithmic}
\end{algorithm}
\section*{SIMULATIONS}
To verify the technical soundness of the proposed methods, simulations were carried out and the results are provided in this section. The spatial distribution is given as an MoG such that
\begin{align*}
    \rho^* &= \sum_{i=1}^{3}\alpha_i\mathcal{N}(\mu_i,\Sigma_i),\\
    \mu_1 &=[80,250]^{T}, \mu_2=[230,60]^T, \mu_3 = [300,310]^T\\
    \Sigma_1 &= \begin{bmatrix}
        15 & 0\\
        0 & 20
    \end{bmatrix},
    \Sigma_2 = \begin{bmatrix}
        30 & 0\\
        0 & 15
    \end{bmatrix},
    \Sigma_3 = \begin{bmatrix}
        15 & 0\\
        0 & 15
    \end{bmatrix}
\end{align*}
where, $\alpha = [\alpha_1, \alpha_2, \alpha_3] = [0.2, 0.3, 0.5]$.

The initial robot position is given as $\mu_0^R = [180,175]^T$ with the covariance matrix for the mass generation in the form of the skinny Gaussian to be $\Sigma^R=\begin{bmatrix}3 & 0\\0 & 3\end{bmatrix}$. The maximum velocity of the robot is limited by $10$.

\begin{figure*}[!t]
    \centering
    \subfloat[]{\includegraphics[scale=0.45]{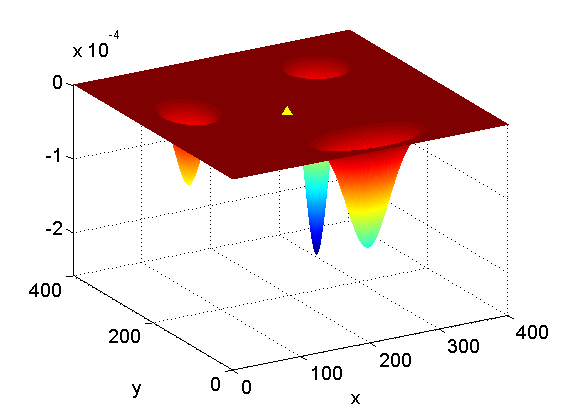}}
    \subfloat[k=10000]{\includegraphics[scale=0.4]{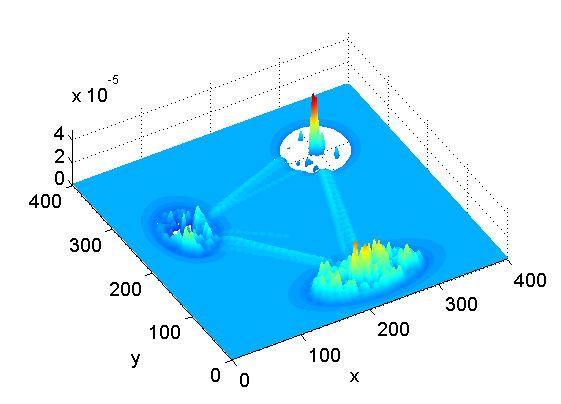}}
    \subfloat[k=50000]{\includegraphics[scale=0.4]{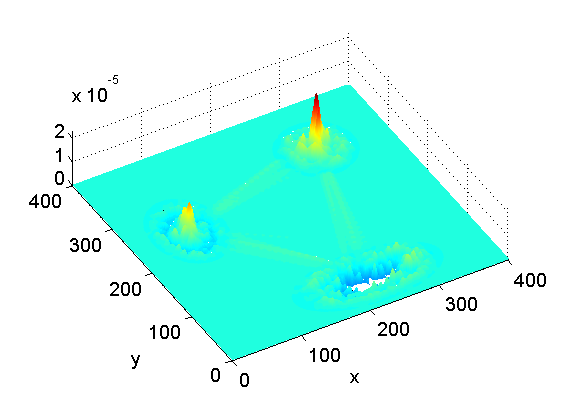}}\\
    \subfloat[k=100000]{\includegraphics[scale=0.4]{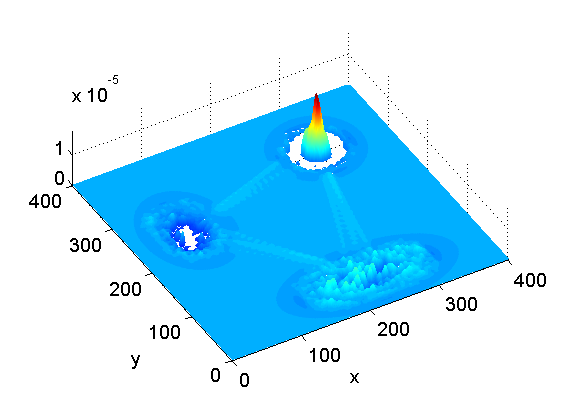}}
    \subfloat[k=150000]{\includegraphics[scale=0.4]{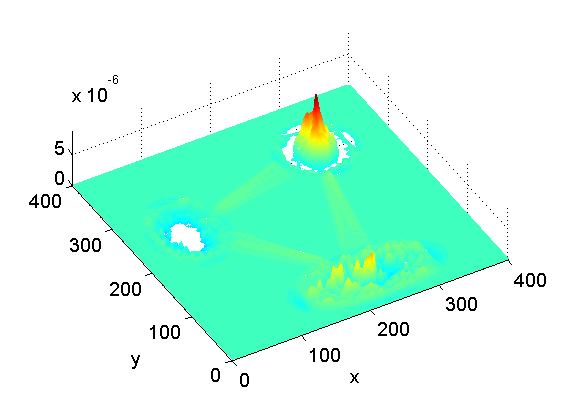}}
    \subfloat[k=200000]{\includegraphics[scale=0.4]{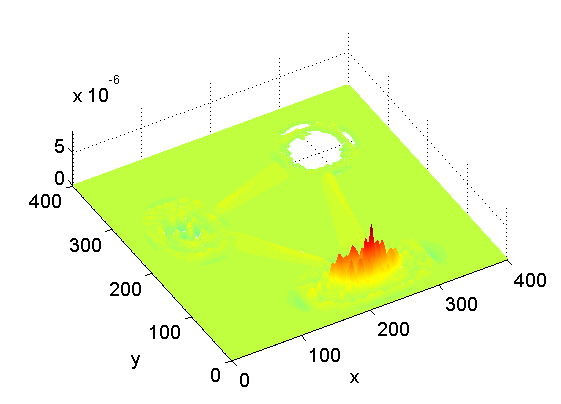}}

    \caption{\uppercase{The spatial distribution with a negative sign and the robot initial position (a) and snapshots of the robot at different time steps (b-f)}}
    \label{fig:snapshot}
\end{figure*}
The spatial distribution with a negative sign and the initial robot position (yellow triangle symbol) are described in Fig. \ref{fig:snapshot} (a). Starting from the initial position, the robot was headed toward the first hole $\mathcal{N}(\mu_1, \Sigma_1)$ since it was the closest one. The robot spent a certain amount of time, $h_1'$, such that the proposed convergence condition \eqref{eqn: theorem_1} and condition for departure time \eqref{eqn: beta} are satisfied. After that, the robot traveled to the second and third holes, and stayed there for the proposed time, which are computed based on \eqref{eqn: theorem_1} and \eqref{eqn: beta}. This is necessary to guarantee the convergence of the ergodic function $V_k$ as time increases. 

In Fig. \ref{fig:error_plot}, $V_k$ vs time plot (top figure) is provided for the convergence result. As shown in this plot, $V_k$ is decreasing as the discrete time $k$ goes up. After a large amount of time, $k=2\times10^5$, the ergodic function reached the value of 0.03, which is evidently small enough to show that the robot attained the ergodicity.
The bottom figure in Fig. \ref{fig:error_plot} indicates the timing for each hole activated as a sub-goal as well as the robot's spending time on it. It is observed that the more time passes, the more robot stays at each hole. This is because the contribution of the mass generation $f_k$ by the robot to $\rho_k$ reduces along with $k$ increment as explained in the last part of the problem description section.  
\begin{figure}[!t]
    \centering
    \includegraphics[scale=0.62]{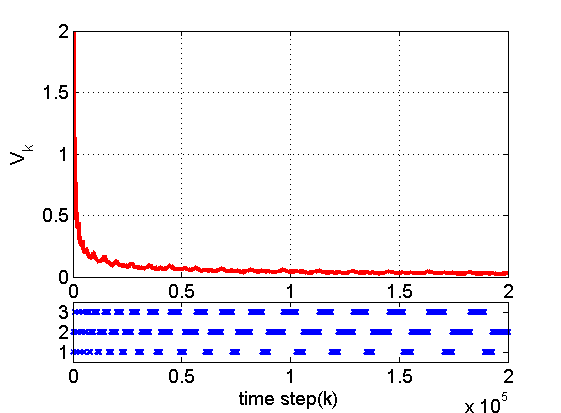}
    \caption{\uppercase{Ergodic function value vs time with the timing diagram of target holes}}
    \label{fig:error_plot}
\end{figure}

Since it may not be clear whether $V_k$ is still decreasing after large time step $k$ (e.g., for $k>10^5$) due to its scale, the log value of $V_k$ is also provided in Fig. \ref{fig:log_error_plot}.  
\begin{figure}[!tbph]
    \centering
    \includegraphics[scale=0.62]{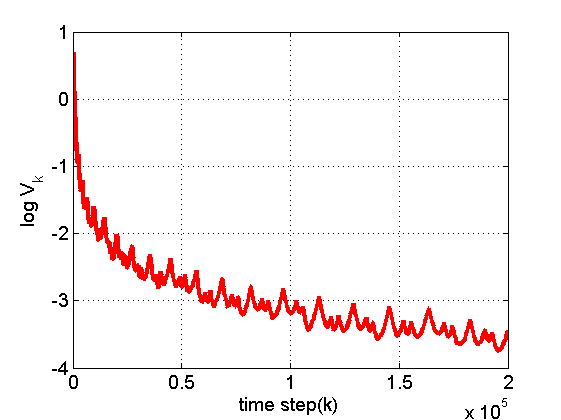}
    \caption{\uppercase{Ergodic function value vs time in log scale}}
    \label{fig:log_error_plot}
\end{figure}
From this plot, $V_k$ keeps decreasing and hence, it can be concluded that the robot will achieve the ergodicity as $k\rightarrow\infty$.
\section*{CONCLUSION}
In this paper, a new approach is proposed to address autonomous robot ergodic exploration problems. 
For this purpose, the mass generation by the robot is assumed to be skinny Gaussian, whereas the spatial distribution as a reference is given by an MoG. 
Differently from the previously developed methods to attain the ergodicity, the proposed one does not include any approximations.
Based on the timing analysis, the convergence condition to achieve the ergodicity is derived.
Also, the formal algorithm to realize the robot ergodic exploration is provided. To verify the proposed methods, simulations were carried out, of which results supports the validity of the proposed convergence result.


\bibliographystyle{asmems4}


%
% ...

\bibliography{reference}



\end{document}